\documentclass[conference]{IEEEtran}
\usepackage{graphicx}
\usepackage{amsfonts,color}
\usepackage[cmex10]{amsmath}
\usepackage{cite}
\usepackage{amssymb}
\usepackage{algorithmic}
\usepackage[center]{caption}
\usepackage{xparse}
\usepackage{epsfig}
\usepackage{latexsym}
\usepackage{epstopdf}
\usepackage{amsmath, amsthm, amssymb}
\usepackage{verbatim}
\usepackage{amsthm,color}
\usepackage[linesnumbered,ruled,vlined]{algorithm2e}

\usepackage{caption}
\usepackage{subcaption}
\usepackage[utf8]{inputenc}
\usepackage{soul}

\begin{document}
%%%%%%%%%%%%%%%%%%%%%%%%%%%%%%%%%%%%%%%%%%%%%%%%%%%%%%%%%%%%%%%%%%%%%%%%%%%%%%%%%%%%%%%%%%%%%%%%%%%%%%%%%%%%%%%%%%%%%%%%%%%%%%%%%%%%%%%%%%%%%%%%%%%%%%%%%%%%%%%%%%%%%%%%%%%%%%%%%%%%%%%%%%%%%%%%%%%%%%%%%%%%%%%%%%%%%%%%%%%%%%%%%%%%%%%%%%%%%%%%%%%%%%%%%%%%%%%%%%%%%%%%%%%%%%%%%%%%%%%%%%%%%%%%%%%%%%%%%%%%%%%%%%%%%%%%%%%%%%%%%%%%%%%%%%%%%%%%%%%%%%%%%%%%%%%%%%%%%%%%%%%%%%%%%%%%%%%%%%%%%%%%%%%%%%%%%%%%%%%%%%%%%%%%%%%%%%%%%%%%%%%%
% Paper Title
\title{Achievable Degrees of Freedom on $K-$user MIMO Multi-way Relay Channel with Common and Private Messages}

\author{\IEEEauthorblockN{Mohamed Salah\IEEEauthorrefmark{2},
        Amr El-keyi\IEEEauthorrefmark{2}, Yahya Mohasseb\IEEEauthorrefmark{2}\IEEEauthorrefmark{1} and
        Mohammed Nafie\IEEEauthorrefmark{3}
        } 
    \IEEEauthorblockA{\IEEEauthorrefmark{2}Wireless Intelligent Networks Center (WINC), Nile University, Cairo, Egypt\\
        \IEEEauthorrefmark{1}Department of Communications, The Military Technical College, Cairo, Egypt\\
         \IEEEauthorrefmark{3}EECE Dept., Faculty of Engineering, Cairo University, Giza, Egypt\\       
        Email: \{\tt m.salah@nileu.edu.eg, aelkeyi@nileuniversity.edu.eg, \{mnafie,mohasseb\}@ieee.org\}
    }
}
%%%%%%%%%%%%%%%%%%%%%%%%%%%%%%%%%%%%%%%%%%%%%%%%%%%%%%%%%%%%%%%%%%%%%%%%%%%%%%%%%%%%%%%%%%%%%%%%%%%%%%%%%%%%%%%%%%%%%%%%%%%%%%%%%%%%%%%%%%%%%%%%%%%%%%%%%%%%%%%%%%%%%%%%%%%%%%%%%%%%%%%%%%%%%%%%%%%%%%%%%%%%%%%%%%%%%%%%%%%%%%%%%%%%%%%%%%%%%%%%%%%%%%%%%%%%%%%%%%%%%%%%%%%%%%%%%%%%%%%%%%%%%%%%%%%%%%%%%%%%%%%%%%%%%%%%%%%%%%%%%%%%%%%%%%%%%%%%%%%%%%%%%%%%%%%%%%%%%%%%%%%%%%%%%%%%%%%%%%%%%%%%%%%%%%%%%%%%%%%%%%%%%%%%%%%%%%%%%%%%%%%%
\maketitle
\begin{abstract}
This paper investigates the achievable total degrees of freedom (DoF) 
of the MIMO multi-way relay channel that consists of $K$ users, where 
each user is equipped with $M$ antennas, and a decode-and-forward
relay equipped with $N$ antennas. In this channel, each user wants
to convey $K-1$ private messages to the other users in addition
to a common message to all of them. Due to the absence of direct
links between the users, communication occurs through the relay in
two phases: a multiple access channel (MAC) phase and a broadcast (BC)
phase. We derive cut-set bounds on the total DoF of the network,
and show that the network has DoF less than or equal to $K\min\{N,M\}$.
Achievability of the upper bound is shown by using
signal space alignment for network coding in the MAC phase, and
zero-forcing precoding in the BC phase. We show that introducing
the common messages besides the private messages leads to
achieving higher total DoF than using the private messages only.
\end{abstract}
%%%%%%%%%%%%%%%%%%%%%%%%%%%%%%%%%%%%%%%%%%%%%%%%%%%%%%%%%%%%%%%%%%%%%%%%%%%%%%%%%%%%%%%%%%%%%%%%%%%%%%%%%%%%%%%%%%%%%%%%%%%%%%%%%%%%%%%%%%%%%%%%%%%%%%%%%%%%%%%%%%%%%%%%%%%%%%%%%%%%%%%%%%%%%%%%%%%%%%%%%%%%%%%%%%%%%%%%%%%%%%%%%%%%%%%%%%%%%%%%%%%%%%%%%%%%%%%%%%%%%%%%%%%%%%%%%%%%%%%%%%%%%%%%%%%%%%%%%%%%%%%%%%%%%%%%%%%%%%%%%%%%%%%%%%%%%%%%%%%%%%%%%%%%%%%%%%%%%%%%%%%%%%%%%%%%%%%%%%%%%%%%%%%%%%%%%%%%%%%%%%%%%%%%%%%%%%%%%%%%%%%%
\IEEEpeerreviewmaketitle
\section{Introduction}
% Acknowledgment
\makeatletter{\renewcommand*{\@makefnmark}{} \footnotetext{ \vspace{0.02in}
This paper was made possible by a grant from the Egyptian National
Telecommunications Regulatory Authority. The statements made herein are solely
responsibility of the authors.}\makeatother}

% General intro
The use of relays has drawn much research interest in wireless networks. 
This is mainly due to its ability to increase the overall throughput, 
improve the energy efficiency, and extend the coverage in the
case of power limited terminals \cite{pabst2004relay}. 
However, the main focus so far is on the mutli-way communication 
where an additional node acting as a relay is supporting the 
exchange of information between a group of users. 
Since characterizing the capacity of the MIMO relay networks 
is too complex, a lot of work was done to study such
networks using an alternative metric, which is the degrees of
freedom (DoF).  

% Literature 1
Recently, Lee \textit{et. al.} studied the DoF of the
symmetric 3-user MIMO Y channel in \cite{lee2010degrees}. 
In this channel, each user equipped with $M$ antennas 
aims to deliver two private messages to the other two 
users via a decode and forward relay equipped with $N$ antennas. 
They showed that if $N \geq 3M/2$, then the 
cut-set bound given by $3M$ DoF is achievable. 
However, the novel contribution is not in achieving the DoF upper 
bound of the network but in their new idea of alignment for network coding. 
They exploit the concept of interference alignment and 
Physical Layer Network Coding (PLNC) to propose a Signal
Space Alignment technique for Network Coding (SSA-NC). 
The main idea of SSA is that each pair of users who want 
to exchange messages cooperatively designs beamforming 
matrices such that the bidirectional vectors of each pair 
of users can be aligned at the relay in the same subspace. 
Then, the relay can simply decodes and forwards the 
linear combination of every pairwise vectors. 
Finally, each user can decode desired vectors by 
subtracting its side information from the network coded received vectors.  
Thus, the phrase of SSA-NC means signal alignment at the relay 
and network coding at users for decoding. 

Based on the work of \cite{lee2010degrees}, 
total DoF of different scenarios has been characterized in 
\cite{chaaban2013degrees,lee2012achievable,tian2012signal,tian2013degrees,liu2014generalized}.
Particularly, the authors in \cite{chaaban2013degrees} investigated 
the total DoF of the 3-user asymmetric MIMO Y channel.
They showed that the general MIMO Y channel with $M_{1} \geq M_{2} \geq
M_{3}$ antennas at user 1, 2, and 3, respectively, and with $N$
antennas at the relay, has $\min \{2M_{2}+2M_{3},M_{1}+M_{2}+M_{3},2N\}$ DoF. 
They derived an outer bound on the DoF using cut-set bounds and genie-aided bounds. 
Then, the work in \cite{lee2012achievable} extended 
the channel model in \cite{lee2010degrees} to the case of $K$
users, where the $j$-th user with $M_j$ antennas sends $K-1$ independent 
messages via the relay equipped with $N$ antennas and each message achieves $d$ DoF. 
They proved the achievability of total DoF $=dK(K-1)$ using SSA-NC.
In \cite{tian2012signal}, the authors studied the DoF of the 4-user relay
network with two clusters of users where each cluster contains two users. 
Each user wants to exchange messages with the other user in the same cluster. 
They achieved a total DoF of $2\min\{2M,N\}$ 
using time division multiple access and SSA-NC, where $M$ and 
$N$ are the number of antennas at each user and the relay, respectively.
Then, they extended this work to the case of $L$-clusters, $K$-user 
MIMO multi-way Relay Channel (mRC) with no direct links, where users in each 
cluster wish to exchange messages within the cluster \cite{tian2013degrees}.  
Finally, the authors of \cite{liu2014generalized} generalized the work of 
\cite{lee2012achievable,chaaban2013degrees,tian2012signal,tian2013degrees} to the case where 
each user can arbitrarily select one or more partners to conduct independent information exchange. 
They showed that the upper bound on the total DoF of 
$\min\{\sum\nolimits_{i=1}^K M_i,2\sum\nolimits_{i=2}^K M_i,2N\}$ 
is tight under the antenna configuration 
$N \geq \max\{\sum\nolimits_{i=1}^K M_i-M_s-M_t+d_{s,t}|\forall s,t\}$, 
where $M_i$ and $N$ are the number of antennas at node $i$ and the relay, 
respectively, and $d_{s,t}$ denotes the DoF of the message exchanged between nodes $s$ and $t$. 

%Our work
In this paper, we consider the network information flow problem 
for a $K$-user Gaussian MIMO mRC and analyze its DoF. 
In this channel, all the $K$ users are equipped with $M$ antennas 
and the relay is equipped with $N$ antennas.
However, unlike most of the prior work where the users were 
assumed to exchange private messages only 
\cite{lee2010degrees, chaaban2013degrees,lee2012achievable,tian2012signal,tian2013degrees,liu2014generalized}, 
we consider both common and private messages.
The $j$-th user intends to convey $K$ messages to the other users by the help 
of an intermediate relay; $K-1$ private messages to the other users in 
addition to a common message to all of them, where $j \in \{1,\cdots, K\}$. 
We first derive cut-set bound on the total DoF, 
and show that the network has $K\min\{N,M\}$ DoF.
By combining PLNC and SSA, we then propose a novel and systematic way of beamforming 
design at the users and at the relay to align signals, and consequently, 
we can efficiently implement PLNC. 
Additionally, we show that our proposed scheme can achieve this upper bound.  
%Finally, we match this result to the case of $K$ users exchanging private messages 
%only \cite{liu2014generalized}, and show that the DoF gain obtained 
%by sending common messages, increases linearly with the number of users.

% Paper sequence
The rest of this paper is organized as follows: In Section
\ref{model}, we describe our network and the main assumptions.
Next, we derive cut-set bounds on the total DoF in Section \ref{sym}. 
Then, we present the proposed achievable scheme in Section \ref{ach}.
Finally, our conclusions are stated in Section \ref{conc}. \\\\
%%%%%%%%%%%%%%%%%%%%%%%%%%%%%%%%%%%%%%%%%%%%%%%%%%%%%%%%%%%%%%%%
% Notations
%%%%%%%%%%%%%%%%%%%%%%%%%%%%%%%%%%%%%%%%%%%%%%%%%%%%%%%%%%%%%%%%
\textit{Notation}

In this work, boldface uppercase letters denote matrices and
boldface lowercase letters are used for vectors. 
$\mathbb{C}$ denotes the complex space. 
For any general matrix $\boldsymbol{G}$, $\boldsymbol{G}^H$,
$\boldsymbol{G}^\dagger$, $\mathcal{S}\{\boldsymbol{G}\}$ and
$\mathcal{N}\{\boldsymbol{G}\}$ denote the Hermitian transpose, the pseudo
inverse, the span of the column vectors, and the null space of $\boldsymbol{G}$,
respectively.
%%%%%%%%%%%%%%%%%%%%%%%%%%%%%%%%%%%%%%%%%%%%%%%%%%%%%%%%%%%%%%%%%%%%%%%%%%%%%%%%%%%%%%%%%%%%%%%%%%%%%%%%%%%%%%%%%%%%%%%%%%%%%%%%%%%%%%%%%%%%%%%%%%%%%%%%%%%%%%%%%%%%%%%%%%%%%%%%%%%%%%%%%%%%%%%%%%%%%%%%%%%%%%%%%%%%%%%%%%%%%%%%%%%%%%%%%%%%%%%%%%%%%%%%%%%%%%%%%%%%%%%%%%%%%%%%%%%%%%%%%%%%%%%%%%%%%%%%%%%%%%%%%%%%%%%%%%%%%%%%%%%%%%%%%%%%%%%%%%%%%%%%%%%%%%%%%%%%%%%%%%%%%%%%%%%%%%%%%%%%%%%%%%%%%%%%%%%%%%%%%%%%%%%%%%%%%%%%%%%%%%%%
\begin{figure}
\begin{subfigure}{.5\textwidth}
  \centering
  \includegraphics[width=.9\linewidth]{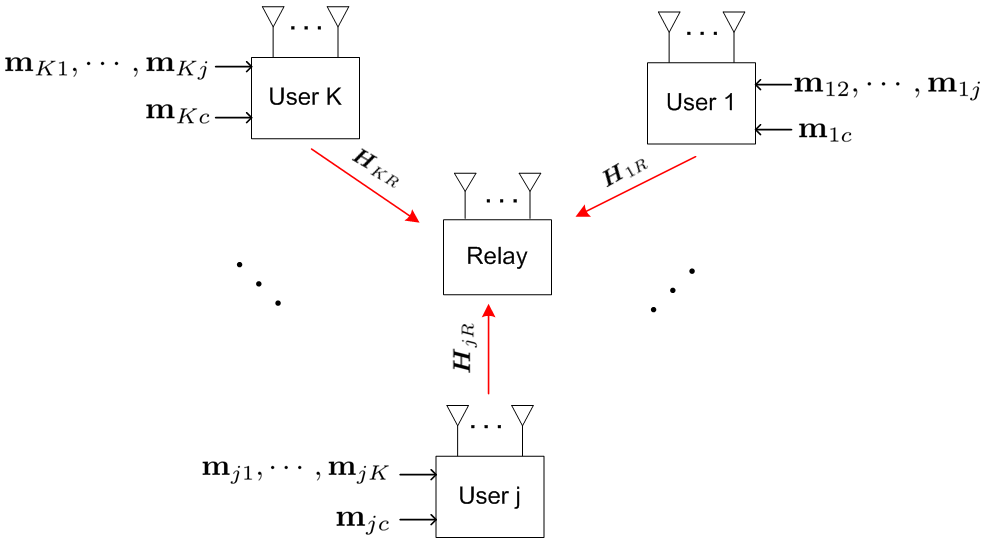}
  \caption{}
  \label{fig:F1}
\end{subfigure}%
\hfill
\\\\
\hrule
\hfill
\\\\
\begin{subfigure}{.5\textwidth}
  \centering
  \includegraphics[width=.9\linewidth]{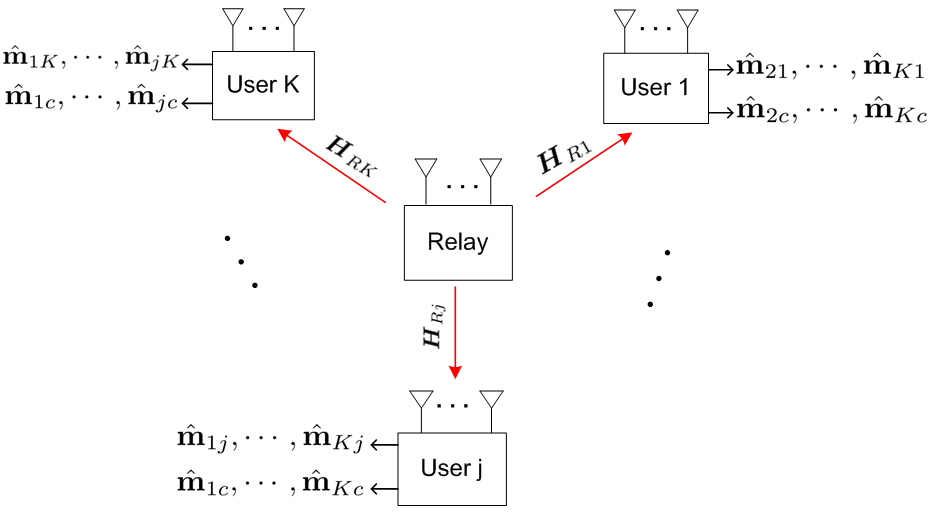}
  \caption{}
  \label{fig:F2}
\end{subfigure}%

\caption{System Model of $K$user MIMO mRC: \{a) MAC phase, (b) Broadcast phase\}}
\label{sys_model}
\end{figure}
\section{System Model}
\label{model}
\subsection{Channel Model}
The system model for $K$-user MIMO mRC is shown in 
\newline Fig. \ref{sys_model}. 
In this channel, $K$ users, each equipped with $M$ antennas, exchange 
messages in a pairwise manner with the help of an intermediate $N$- 
antennas relay node. Specifically, user $j$ wants to send $K$ messages; 
$K-1$ private messages $\mathbf{m}_{ji}$ to user $i$ in addition to 
a common message $\mathbf{m}_{jc}$ to the remaining $K-1$ users.
Also, the $\textit{j}$-th user intends to decode all the other users' 
messages; $K-1$ private messages $\hat{\mathbf{m}}_{ij}$ in addition 
to $K-1$ common messages $\hat{\mathbf{m}}_{ic}$,where $i,j \in \{1,\cdots,K\}, i \neq j $.
It is assumed that the users can communicate only through the relay 
and no direct links exist between any pairs of users. 
Additionally, we assume that the channel between each user 
and the relay is reciprocal.
The transmission takes place over two phases: MAC and BC. 
In the MAC phase, each user transmits its messages to the relay. 
The received signal at the relay is given by
\begin{equation}
  \boldsymbol{\mathrm{y}}_{r} = \sum\limits_{j=1}^K \boldsymbol{H}_{jR}\boldsymbol{\mathrm{x}}_{j}
  + \boldsymbol{\mathrm{z}}_{r}
\end{equation}
where $\boldsymbol{\mathrm{y}}_{r}$ and $\boldsymbol{\mathrm{z}}_{r}$ 
denote the $ N \times 1 $ received signal vector and 
additive white Gaussian noise (AWGN) vector with unit variance at the relay, respectively. 
$\boldsymbol{\mathrm{x}}_{j}$ is an $M \times 1$
vector representing the transmitted signal from user $j$.
$\boldsymbol{H}_{jR}$ is the $N \times M$ channel matrix 
from user $j$ to the relay, where the entries of the 
channel matrices $\boldsymbol{H}_{jR}$, for $ j \in \{1,\cdots,K\}$, 
are independently and identically distributed zero mean 
complex Gaussian random variables with unit variance. 
Thus, we guarantee that each of the channel matrices is full rank almost surely. 
In the BC phase, the received signal at user $j$ from the relay is given by
\begin{equation}
  \boldsymbol{\mathrm{y}}_{j} = \boldsymbol{H}_{Rj}\boldsymbol{\mathrm{x}}_{r}
  + \boldsymbol{\mathrm{z}}_{j}
\end{equation}
where $\boldsymbol{\mathrm{y}}_{j}$ and $\boldsymbol{\mathrm{z}}_{j}$ 
are the $ M \times 1 $ received vector and AWGN vector at the 
$j$-th user, respectively.
$\boldsymbol{\mathrm{x}}_{r}$ is an $N \times 1$ vector representing
the transmitted signal vector from the relay.
Finally, $\boldsymbol{H}_{Rj}$ denote the $M \times N$ channel matrix 
from the relay to user $j$. 
Throughout this paper, we assume 
that perfect channel state information is available at all users and at 
the relay in both phases. Also, we assume that all the users have the same power constraint. 
Moreover, all nodes are assumed to operate in full-duplex mode. \footnote{All the DoF results can be applied to the half-duplex case with only a multiplicative factor of $\dfrac{1}{2}$}

%%%%%%%%%%%%%%%%%%%%%%%%%%%%%%%%%%%%%%%%%%%%%%%%%%%%%%%%%%%%%%%%%%%%%%%%%%%%%%%%%%%%%%%%%%%%%%%%%%%%%%%%%%
%%%%%%%%%%%%%%%%%%%%%%%%%%%%%%%%%%%%%%%%%%%%%%%%%%%%%%%%%%%%%%%%%%%%%%%%%%%%%%%%%%%%%%%%%%%%%%%%%%%%%%%%%%
\subsection{Degrees of freedom}
In the absence of precise capacity characterizations, researchers
have pursued asymptotic and/or approximate capacity characterizations. 
One of the most important capacity characterizations is the DoF of the
network \cite{jafar2008degrees,cadambe2008degrees,jafar2007degrees}. 
The DoF metric is primarily concerned with the
limit where the total transmit power approaches infinity, while the values
of channel coefficients and the local noise power remain unchanged.
Thus, if we denote by $C(P)$ the sum capacity with the total transmit
power $P$, then the DoF metric $\eta$ is defined as
\begin{equation*}
\eta = \lim_{P \to \infty} \frac{C(P)}{\log(P)}
\label{DoF}
\end{equation*}

In other words, we can consider the DoF as the number of signaling dimensions,
where 1 signal dimension corresponds to one interference-free
AWGN channel with SNR increasing proportionately 
with $P$ as $P$ approaches infinity. 
Let $d_{ij}$ denote the number of
interference-free streams sent from user $i$ to user $j$. Also, we
define $d_{jc}$ as the number of interference-free streams sent from user
$j$ to the remaining $K-1$ users. 
The total DoF of the MIMO mRC with private and common messages is 
given by the total number of interference-free streams sent through the network. 
As a result, the total DoF $D_{t}$ is equal to
  \begin{equation}
    D_{t} = \sum\limits_{j=1}^K \sum\limits_{i=1,i \neq j}^K d_{ji}
     +(K-1)\sum\limits_{j=1}^K d_{jc}
 \end{equation}
Notice that in the above equation the contribution of the common
messages to the total DoF is weighted by a factor $K-1$ as each common
message is directed to $K-1$ users.
%%%%%%%%%%%%%%%%%%%%%%%%%%%%%%%%%%%%%%%%%%%%%%%%%%%%%%%%%%%%%%%%%%%%%%%%%%%%%%%%%%%%%%%%%%%%%%%%%%%%%%%%%%%%%%%%%%%%%%%%%%%%%%%%%%%%%%%%%%%%%%%%%%%%%%%%%%%%%%%%%%%%%%%%%%%%%%%%%%%%%%%%%%%%%%%%%%%%%%%%%%%%%%%%%%%%%%%%%%%%%%%%%%%%%%%%%%%%%%%%%%%%%%%%%%%%%%%%%%%%%%%%%%%%%%%%%%%%%%%%%%%%%%%%%%%%%%%%%%%%%%%%%%%%%%%%%%%%%%%%%%%%%%%%%%%%%%%%%%%%%%%%%%%%%%%%%%%%%%%%%%%%%%%%%%%%%%%%%%%%%%%%%%%%%%%%%%%%%%%%%%%%%%%%%%%%%%%%%%%%%%%%
\section{Upper Bounds on the Degrees of Freedom }
\label{sym}
In this section, we provide upper bounds on the total DoF
of the Gaussian MIMO mRC with common and private messages. 
One way to obtain upper bounds for the $K$-user MIMO mRC 
is by using the cut-set bounds. 
In \cite{lee2010degrees}, Lee \textit{et al}. derived cut set
bounds for the Gaussian MIMO Y channel with equal number of
antennas at all users, where the users exchange private messages only. 
We introduce the common messages to the network and generalize 
these bounds for arbitrary number of users. 
Then, using the cut set-bounds, we can obtain the DoF bounds.
The main result of this section is stated in the following theorem.
\newtheorem{theorem}{Theorem}
\begin{theorem}
\label{the1} The total number of DoF of a Gaussian MIMO mRC is given by
\begin{equation}
     D_{t} =  K\min \{N,M\}
    \label{UB}
\end{equation}
\end{theorem}
%\subsection{Proof of Theorem 1} \label{ACH}
\begin{proof}
 
Let us label the set of nodes in the network by $\mathcal{P} \triangleq 
\{U_1,\cdots,U_K,R\}$, where $U_j$ denotes user $j$, for $j \in \{1,\cdots,K\}$, 
and $R$ denotes the relay. Moreover, we define the two sub-sets $ \mathcal{T}$  and $\mathcal{T}^{c}$
as the transmitting and receiving nodes, respectively, 
where $\mathcal{T}$ and $\mathcal{T}^{c} \in \mathcal{P}$. 
By considering only the messages intended to $U_1$ 
from $\{U_2,\cdots,U_K\}$ and writing the cut in the MAC phase 
where $\mathcal{T} = \{U_2,\cdots,U_K\}$ and $\mathcal{T}^{c} {=} \{R,U_1\}$, and 
applying the cut-set theorem in \cite{cover2012elements} to this cut leads to
the following inequality
\begin{equation}
\sum\limits_{j=2}^K R_{j1} + \sum\limits_{j=2}^K R_{jc} \leq  I(\boldsymbol{\mathrm{x}}_2,\cdots,
   \boldsymbol{\mathrm{x}}_K;\boldsymbol{\mathrm{y}}_r \mid \boldsymbol{\mathrm{x}}_1)
\end{equation}
During the BC phase, the messages from  $\{U_2,\cdots,U_K\}$ will be relayed 
to $U_1$, 
we consider $\mathcal{T} = \{R,2,\cdots,K\}$ and $\mathcal{T}^{c} = \{U_1\}$, we get 
\begin{equation}
\sum\limits_{j=2}^K R_{j1} + \sum\limits_{j=2}^K R_{jc} \leq I(\boldsymbol{\mathrm{x}}_r;
\boldsymbol{\mathrm{y}}_1)
\end{equation}
where $I(\boldsymbol{\mathrm{x;y}})$ is the mutual information
between random variables $\boldsymbol{\mathrm{x}}$ and
$\boldsymbol{\mathrm{y}}$, $R_{ij}$ is the maximum information
transfer rate from $U_i$ to $U_j$, $i \neq j$, and $R_{jc}$ is the
maximum information transfer rate from $U_j$ to the remaining 
$K-1$ users. Using (5) and (6), we can bound the DoF as follows
\begin{eqnarray}
 \sum\limits_{j=2}^K d_{j1}+d_{jc} &\leq& \min \{N,(K-1)M\} \\
 \sum\limits_{j=2}^K d_{j1}+d_{jc} &\leq& \min \{N,M\}
\end{eqnarray}
Combining (7) and (8), we get
\begin{equation}
 \sum\limits_{j=2}^K d_{j1}+d_{jc} \leq \min \{N,M\}
\end{equation}
Similarly, by considering the messages intended to $\{U_2,\cdots,K\}$, respectively, 
and taking the appropriate cuts during MAC and BC phases, we get $K-1$ 
equations similar to (9) but with different intended user, By adding 
the $K$ equations we obtain
 \begin{equation}
 D_{t} \leq  K\min \{N,M\}
 \end{equation}
 
%\begin{figure*}
%\begin{subfigure}{.5\textwidth}
%  \centering
%  \includegraphics[width=80mm,scale=2]{Fig_2a.png}
%  \caption{}
%  \label{fig2a}
%\end{subfigure}%
%\begin{subfigure}{.5\textwidth}
%  \centering
%  \includegraphics[width=90mm,scale=2]{Fig_2b.png}
%  \caption{}
%  \label{fig2b}
%\end{subfigure}%
%\caption{Achievable Total DoF a) mRC with $K = 3$ b) mRC with general $K$}
%\label{fig2}
%\end{figure*}
\end{proof} 
\section{Achievability} \label{ach}
In this section, we describe the transmission
strategies that can achieve the optimal DoF mentioned in Theorem \ref{the1}. 
We split our proof into two cases according to the relationship
between $M$ and $N$. Our achievability schemes depend mainly on
finding linearly independent subspaces in which each pair of users
can align their bi-directional vectors. 

\subsubsection{Case 1: $ N < M$}\label{C1}
In this case, $KN$ is the dominant term in the RHS of (53). In the
MAC phase, we design the beamforming matrices at the transmitting
users such that the signal vector coming from each user is
partially aligned at the relay with that transmitted from user
$1$. We use only the common messages in our achievability proof in
this case. Thus, the transmitted signals of the users in the
uplink are given by
 \begin{eqnarray}
          \boldsymbol{\mathrm{x}}_{1} &{=}& \sum\limits_{j=2}^K  \boldsymbol{V}^j_{1c}
          \boldsymbol{\mathrm{s}}_{1c} \\
          \boldsymbol{\mathrm{x}}_{j} &{=}& \boldsymbol{V}_{jc} \boldsymbol{\mathrm{s}}_{jc}, \quad
          2 \leq j \leq K
 \end{eqnarray}
where $\boldsymbol{\mathrm{s}}_{jc}$ is $N/(K-1) \times 1$ data
vector, for $j \in \{1,\cdots,K\}$, containing common messages to
be delivered from user $j$ to the remaining $K-1$ users. The
beamforming matrices $\boldsymbol{V}^j_{1c}$ and
$\boldsymbol{V}_{jc}$ of dimensions $M \times N/(K-1)$, for $j \in
\{2,\cdots,K\}$, are chosen such that
\begin{equation}
   \mathcal{S}\{\boldsymbol{H}_{1R}\boldsymbol{V}^j_{1c}\} =
   \mathcal{S}\{\boldsymbol{H}_{jR}\boldsymbol{V}_{jc}\}, \quad 2 \leq j \leq K
 \end{equation}
where the beamforming matrices at user 1 are chosen randomly,
and the rest are chosen as follows
\begin{equation}
     \boldsymbol{V}_{jc} = \boldsymbol{H}_{jR}^\dagger
     \boldsymbol{H}_{1R}\boldsymbol{V}^j_{1c}.
\end{equation}

The relay will receive a superposition of the transmitted signals
from all the network users. Thus, the received signal at the relay
is given by
\begin{equation}
  \boldsymbol{\mathrm{y}}_{r} = \sum\limits_{j=2}^K \boldsymbol{H}_{1R} \boldsymbol{V}^j_{1c}(\boldsymbol{\mathrm{s}}_{1c}+\boldsymbol{\mathrm{s}}_{jc})
   + \boldsymbol{\mathrm{z}}_{r}.
\end{equation}
Let $\boldsymbol{\mathrm{w}}_{1j}$ where $j \in \{2,\cdots,K\}$ be
the $N/(K-1) \times 1$ vector which contains a noisy linear
combination of $\boldsymbol{\mathrm{s}}_{1c}$ and
$\boldsymbol{\mathrm{s}}_{jc}$. The relay can obtain
$\boldsymbol{\mathrm{w}}_{1j} =
\boldsymbol{F}_{1j}^H\boldsymbol{\mathrm{y}}_r$, where
$\boldsymbol{F}_{1j} \in \mathbb{C}^{N \times N/(K-1)}$ is the
zero-forcing matrix used by the relay to remove the undesired
terms. Thus the relay obtains
\begin{equation}
   \boldsymbol{\mathrm{w}}_{1j} = \boldsymbol{F}_{1j}^{H}\boldsymbol{H}_{1R}\boldsymbol{V}^2_{1c}(\boldsymbol{\mathrm{s}}_{1c}+\boldsymbol{\mathrm{s}}_{jc})
   + \boldsymbol{F}_{1j}^{H}\boldsymbol{\mathrm{z}}_{r}.
\end{equation}

Since all the decoded vectors at the relay contain common
messages, these vectors should be transmitted to all the users in
the downlink. In the broadcast phase, the relay transmits
\begin{equation}
  \boldsymbol{\mathrm{x}}_{r} = \sum\limits_{j=2}^K \boldsymbol{T}_{1j}\boldsymbol{\mathrm{w}}_{1j}
\end{equation}
where $\boldsymbol{T}_{1j}$'s  are randomly chosen precoding
matrices, of dimensions $N \times N/(K-1)$. Now, we consider the
received signal at users $j$ which is given by
\begin{equation}
 \boldsymbol{\mathrm{y}}_{j} = \boldsymbol{H}_{Rj}\sum\limits_{i=2}^K \boldsymbol{T}_{1i}\boldsymbol{\mathrm{w}}_{1i}
 + \boldsymbol{\mathrm{z}}_{j}.
 \end{equation}
We use zero-forcing at user $j$ to get
$\boldsymbol{\mathrm{w}}_{1j}$, for $j \in \{2,\cdots,K\}$. For
example, let $\boldsymbol{N}_{12} \in \mathbb{C}^{M \times
N/(K-1)}$ be the zero-forcing matrix used by user 2 such that
\begin{equation}
 \boldsymbol{N}_{12} \subseteq  \mathcal{N}\{
 (\boldsymbol{H}_{R2}\boldsymbol{T}_{13},\cdots,\boldsymbol{H}_{R2}\boldsymbol{T}_{1K})^{T}\}.
 \end{equation}
As a result, user 2 can obtain
$\boldsymbol{\mathrm{w}}_{12}=\boldsymbol{N}^{H}_{12}\boldsymbol{\mathrm{y}}_{2}$.
Similarly, user 2 can obtain $\boldsymbol{\mathrm{w}}_{1j} =
\boldsymbol{N}^{H}_{1j}\boldsymbol{\mathrm{y}}_{1}$, for $j \in
\{3,\cdots,K\}$. Now, user 2 can simply decode
$\boldsymbol{\mathrm{s}}_{1c}$ by removing the contribution of
$\boldsymbol{\mathrm{s}}_{2c}$ from
$\boldsymbol{\mathrm{w}}_{12}$, then decodes
$\boldsymbol{\mathrm{s}}_{jc}$ by subtracting
$\boldsymbol{\mathrm{s}}_{1c}$ from $\boldsymbol{\mathrm{w}}_{1j}$
achieving $N$ DoF. Finally, user 1 recovers
$\boldsymbol{\mathrm{s}}_{jc}$ by subtracting
$\boldsymbol{\mathrm{s}}_{1c}$ from
$\boldsymbol{\mathrm{w}}_{1j}$, for $2 \leq j \leq K$, achieving
$N$ DoF. Thus, this scheme achieves a total of $KN$ DoF which is
equal to the upper bound. It should be noted that if $N$ is not 
divisible by $K-1$, we consider a $(K-1)$-time-slot-symbol extension of
the channel \cite{jafar2011interference}, and then proceed with designing 
the transmit strategy as explained above.

\subsubsection{Case 2: $ N \geq M$ }
In this case, {$D_t = KM$}. We shut down $N-M$ antennas at the
relay and use the same achievability scheme of case 1 where each user sends $M/(K-1)$ common messages to the remaining $K-1$ users. \\
This completes the proof of the achievability of Theorem
\ref{the1}.
\newtheorem*{corollary}{Corollary}
\begin{corollary}
The total DoF of the K-user MIMO mRC with private and common messages is
greater than the total DoF achievable with private messages only.
%equal that with private messages only.
\end{corollary}
\begin{proof}
Table \ref{table:dof_gain} compares the total DoF of the MIMO mRC
in both cases.
\end{proof}
\begin{table*}\footnotesize
\caption{\\\textsc{Achievable Total DoF}}
\label{table:dof_gain}
\centering
\begin{tabular}{|c|c|c|c|}
  \hline
  \textbf{Case} & \textbf{Configuration} & \textbf{mRC with Private messages only} & 
   \textbf{mRC with Private and Common messages}\\
  \hline
  1 & $ N {\leq} M$ & $D_t=2N$ & $D_t=KN$ \\
  \hline
  2 & $ M \leq N \leq \dfrac{2K^2-2K}{K^2-K+2}M$ & $D_t=2N$ & $D_t=KM$ \\
  \hline
  3 & $ \dfrac{2K^2-2K}{K^2-K+2}M \leq N \leq \dfrac{K}{2}M$  & $D_t=\dfrac{4K^2-4K}{K^2-K+2}M$ & $D_t=KM$ \\
  \hline
 
  4 & $ \dfrac{K}{2}M \geq N \geq \dfrac{K^2-3K+3}{K-1}M$  & $D_t= \dfrac{K(K-1)}{K^2-K+2}[2N+(4-K)]M $ & $D_t= KM$\\
  
  \hline 
  5 & $  N \geq \dfrac{K^2-3K+3}{K-1}M $  & $D_t=KM$ & $D_t=KM$\\
  \hline
\end{tabular}
\label{T1}
\end{table*}

In fact, the limited number of antennas at the relay always represented a 
bottleneck in achieving higher DoF.  
This problem becomes more noticeable as the number of users in the network increases, 
especially when direct links are missing. 
For example, when private messages only are used \cite{liu2014generalized},
the respective DoF bound can be easily achieved when $K=3$. 
However, as the number of users increases, i.e. $K>3$, the 
DoF upper bound cannot be achieved using private messages only when 
$\dfrac{K}{2} \geq \dfrac{N}{M} \geq \dfrac{K^2-3K+3}{K-1}$ as Table \ref{table:dof_gain} depicts.
We conjecture that the number of antennas at the relay in this region are 
not sufficient for achieving the $KM$ DoF bound.  
 
On the other hand,  Table \ref{table:dof_gain} shows that introducing the common messages
enables us to achieve up to $KN$ DoF, where the upper bound can be achieved 
by using the common messages only. 
This can be attributed to the fact that each common message 
is directed to $K-1$ users, and hence, its contribution to the 
total DoF scales up with increasing the number of users in the network.
Additionally, the table shows that the gain obtained due to the existence 
of common messages increases linearly with $K$, when $N \leq M$.  

\section{CONCLUSION} \label{conc} 
The K-user MIMO mRC with private and common messages
and with equal number of antennas at the users was studied. 
In this network, each user aims to exchange $K$ messages with 
the other $K-1$ users via a common relay; 
$K-1$ private messages to the other users 
in addition to a common message to all of them. 
We derived cut-set bounds on the total DoF, and showed that
this upper bound, $K\min\{N,M\}$, is achievable.
Thus, we provided a simple way of aligning signals at the users 
and at the relay.  

We compared between the achievable DoF obtained in our network and 
those obtained from the one containing private messages only. 
Consequently, we showed that introducing common messages besides the
private messages can lead to achieving higher DoF gain. 
Moreover, we showed that this gain increases linearly with 
the number of users in the network.  
Additionally, cut-set bounds of different networks can be 
achieved using private and common messages. 
An interesting direction for future study would be to
investigate the total DoF of multicast networks, where each user
can send common messages to a group of users in addition to
private messages.

%% Citation
\bibliographystyle{IEEEtran}
\bibliography{IEEEabrv,Draft_1}

\begin{thebibliography}{10}
\providecommand{\url}[1]{#1}
\csname url@samestyle\endcsname
\providecommand{\newblock}{\relax}
\providecommand{\bibinfo}[2]{#2}
\providecommand{\BIBentrySTDinterwordspacing}{\spaceskip=0pt\relax}
\providecommand{\BIBentryALTinterwordstretchfactor}{4}
\providecommand{\BIBentryALTinterwordspacing}{\spaceskip=\fontdimen2\font plus
\BIBentryALTinterwordstretchfactor\fontdimen3\font minus
  \fontdimen4\font\relax}
\providecommand{\BIBforeignlanguage}[2]{{%
\expandafter\ifx\csname l@#1\endcsname\relax
\typeout{** WARNING: IEEEtran.bst: No hyphenation pattern has been}%
\typeout{** loaded for the language `#1'. Using the pattern for}%
\typeout{** the default language instead.}%
\else
\language=\csname l@#1\endcsname
\fi
#2}}
\providecommand{\BIBdecl}{\relax}
\BIBdecl

\bibitem{pabst2004relay}
R.~Pabst, B.~Walke, D.~Schultz, P.~Herhold, H.~Yanikomeroglu, S.~Mukherjee,
  H.~Viswanathan, M.~Lott, W.~Zirwas, M.~Dohler, H.~Aghvami, D.~Falconer, and
  G.~Fettweis, ``Relay-based deployment concepts for wireless and mobile
  broadband radio,'' \emph{IEEE Communications Magazine}, vol.~42, no.~9, pp.
  80--89, Sept. 2004.

\bibitem{lee2010degrees}
N.~Lee, J.-B. Lim, and J.~Chun, ``Degrees of freedom of the{ MIMO Y} channel:
  Signal space alignment for network coding,'' \emph{IEEE Transactions on
  Information Theory}, vol.~56, no.~7, pp. 3332--3342, July 2010.

\bibitem{chaaban2013degrees}
A.~Chaaban, K.~Ochs, and A.~Sezgin, ``The degrees of freedom of the {MIMO
  Y}-channel,'' in \emph{IEEE International Symposium on Information Theory
  (ISIT)}, Istanbul, Turkey, July 2013, pp. 1581--1585.

\bibitem{lee2012achievable}
K.~Lee, N.~Lee, and I.~Lee, ``Achievable degrees of freedom on {K}-user {Y}
  channels,'' \emph{IEEE Transactions on Wireless Communications}, vol.~11,
  no.~3, pp. 1210--1219, Feb. 2012.

\bibitem{tian2012signal}
Y.~Tian and A.~Yener, ``Signal space alignment and degrees of freedom for the
  two-cluster multi-way relay channel,'' in \emph{1st IEEE Int Conf on Comm in
  China (ICCC)}, Aug. 2012, pp. 12--17.

\bibitem{tian2013degrees}
------, ``Degrees of freedom for the {MIMO} multi-way relay channel,'' in
  \emph{IEEE International Symposium on Information Theory (ISIT)}, Istanbul,
  Turkey, July 2013, pp. 1576--1580.

\bibitem{liu2014generalized}
K.~Liu, M.~Tao, and D.~Yang, ``Generalized signal alignment for arbitrary
  {MIMO} two-way relay channels,'' in \emph{Proc. of IEEE Global
  Telecommunication Conference}, Austin, Tx, Dec. 2014, pp. 1684--1689.

\bibitem{jafar2008degrees}
S.~A. Jafar and S.~Shamai, ``Degrees of freedom region of the {MIMO} {X}
  channel,'' \emph{IEEE Transactions on Information Theory}, vol.~54, no.~1,
  pp. 151--170, Jan. 2008.

\bibitem{cadambe2008degrees}
V.~R. Cadambe and S.~A. Jafar, ``Degrees of freedom of wireless {X} networks,''
  in \emph{IEEE International Symposium on Information Theory}, July 2008, pp.
  1268--1272.

\bibitem{jafar2007degrees}
S.~A. Jafar and M.~J. Fakhereddin, ``Degrees of freedom for the {MIMO}
  interference channel,'' \emph{IEEE Transactions on Information Theory},
  vol.~53, no.~7, pp. 2637--2642, July 2007.

\bibitem{cover2012elements}
T.~M. Cover and J.~A. Thomas, \emph{Elements of information theory}.\hskip 1em
  plus 0.5em minus 0.4em\relax John Wiley \& Sons, 2012.

\bibitem{jafar2011interference}
S.~A. Jafar, ``Interference alignment: {A} {N}ew {L}ook at {S}ignal
  {D}imensions in a {C}ommunication {N}etwork,'' \emph{Foundations and Trends
  in Communications and Information Theory}, vol.~7, no.~1, pp. 1--136, 2011.

\end{thebibliography}

\end{document}